\documentclass[12pt]{article}
\usepackage{amsmath, amssymb, amsthm, mathrsfs}
\usepackage{enumerate}
\usepackage{colordvi}

\newcommand{\e}{\mathrm{e}}
\newcommand{\R}{\mathbb{R}}

\newtheorem{claim}{Claim}[section]
\newtheorem{theorem}[claim]{Theorem}
\newtheorem{corollary}[claim]{Corollary}

\usepackage{cite}

\begin{document}
\begin{center}
{\Large{\textbf{A regular version of Smilansky model}}}

\bigskip

{\large{Diana Barseghyan$^a$, Pavel Exner$^{b,c}$}}

\end{center}

\begin{quote}

{\small a) Fachbereich Mathematik, Universit\"at Erlangen-N\"urnberg, \\ \phantom{c)} Cauerstra▀e 11,
91058 Erlangen, Germany
 \\
b) Doppler Institute for Mathematical Physics and Applied \\
\phantom{c)} Mathematics, Czech Technical University, B\v{r}ehov\'{a} 7, 11519 Prague \\
c) Nuclear Physics Institute ASCR, 25068 \v{R}e\v{z} near Prague, Czechia \\
\phantom{c)} \emph{barseghyan@math.fau.de, exner@ujf.cas.cz}}

\end{quote}

\bigskip

\textbf{Abstract.} We discuss a modification of Smilansky model in which a singular potential `channel' is replaced by a regular, below unbounded potential which shrinks as it becomes deeper. We demonstrate that, similarly to the original model, such a system exhibits a spectral transition with respect to the coupling constant, and determine the critical value above which a new spectral branch opens. The result is generalized to situations with multiple potential `channels'.

\bigskip

\textbf{Mathematical Subject Classification (2010).} 81Q10, 35J10

\bigskip

\textbf{Keywords.} Smilansky model, Schr\"odinger operators, spectral transition

%%%%%%%%%%%%%%%%%%%%%%%%%%%%%%%%%%%%%%%
\section{Introduction} \label{s: intro}
\setcounter{equation}{0}

In the seminal paper \cite{Sm04} Uzy Smilansky discussed a simple example of quantum dynamics which could exhibit a behavior one can regard as irreversible. The model in which it can be demonstrated allows for interpretation in different ways, as a one-dimensional system coupled to a heat bath, as a particular quantum graph, or as a two-dimensional quantum system described by the Hamiltonian
 % ------------- %
\begin{equation}
\label{HSmil}
H_\mathrm{Sm}=-\frac{\partial^2}{\partial x^2} +\frac12\left( -\frac{\partial^2}{\partial y^2}+y^2 \right)
+\lambda y\delta(x)\,.
\end{equation}
 % ------------- %
It was argued in \cite{Sm04} that the behavior of the system depends crucially on the coupling parameter: if $|\lambda|>1$ the particle can escape to infinity along the singular `channel' in the $y$ direction.

The claim can be made mathematically rigorous in terms of the spectral properties of such an operator: one can prove that that for $|\lambda|$ exceeding the critical value the operator has an additional branch of absolutely continuous spectrum which is not bounded from below \cite{So04}. The model was subsequently generalized to the situation when one has more than one singular `channel' ---  cf.~\cite{NS06, ES05} --- and its further properties were studied, in particular, the discrete spectrum in the subcritical case.

It has appeared that there is also another motivation to investigate such systems. Recently Guarneri has used the model --- or rather its modification in which the motion in the $x$ direction is restricted to a finite interval with periodic boundary conditions --- to describe quantum measurements \cite{Gu11}; he studied the time evolution in such a situation identifying the escape along a particular `channel' with reduction of the wave packet.

The paper \cite{Gu11} concludes with expressing the hope that `similar behavior may be reproducible with smoother interaction potentials and also in purely classical models'. The aim of the present paper is demonstrate that this is indeed the case. We are going to investigate a model in which the $\delta$ interaction with $y$-dependent strength is replaced by a smooth potential channel of increasing depth, and to show that it exhibits the analogous spectral transition as the coupling parameter exceeds a critical value.

Replacing the $\delta$ interaction by a regular potentials, however,  requires modifications, in particular, the coupling cannot be linear in $y$ and the profile of the channel has to change with $y$; in this respect our present problem is similar to another model we have investigated recently \cite{EB12}. To understand the reason one should realize that the essence of the effect lays in the fact that far from the $x$-axis the variables in the solution to the Schr\"odinger equation effectively decouple --- one can regard it as a sort of adiabatic approximation --- and the oscillator potential competes with the principal eigenvalue of the `transverse' part of the operator, which in the singular case equals to $\frac14 \lambda^2y^2$. If we want to approximate the $\delta$ interaction by a family of shrinking potential in the usual way \cite[Sec.~I.3.2]{AGHH05} we have to match the integral of the potential with the $\delta$ coupling constant, $\int U(x,y)\,\mathrm{d}x \sim y$, which can be achieved, e.g., by choosing $U(x,y)=\lambda y^2V(xy)$ for a fixed function $V$.

Inspired by these considerations we are going to investigate the model described by the partial differential operator on $L^2(\mathbb{R}^2)$ acting as
 % ------------- %
\begin{equation}
\label{H1}
H=-\frac{\partial^2}{\partial x^2}-\frac{\partial^2}{\partial y^2}+\omega^2y^2
-\lambda y^2V(xy)\chi_{\{|x|\le a\}}(y),
\end{equation}
 % ------------- %
where $\omega,\,a$ are positive constants, $\chi_{\{|y|\le a\}}$ is the indicator function of the interval $(-a,a)$, and the potential $V$ with $\mathrm{supp}\,V\subset[-a,a]$ is a nonnegative function with bounded first derivative. By Faris-Lavine theorem \cite[Thms.~X.28 and X.38]{RS} the above operator is essentially self-adjoint on $C_0^\infty (\mathbb{R}^2)$; the same is true for its generalization,
 % ------------- %
\begin{equation}
\label{HN}
H=-\frac{\partial^2}{\partial x^2}-\frac{\partial^2}{\partial y^2}+\omega^2y^2
- \sum_{j=1}^N \lambda_j y^2V_j(xy)\chi_{\{|x-b_j|\le a_j\}}(y)
\end{equation}
 % ------------- %
with a finite number of potential channels, where the functions $V_j$ are positive with bounded first derivative, with the supports contained in the intervals $(b_j-a_j,b_j+a_j)$ and such that $\mathrm{supp}\,V_j \cap \mathrm{supp}\,V_k = \emptyset$ holds for $j\ne k$.

Our aim in the present paper is to demonstrate existence of a critical coupling separating two different situations: below it the spectrum is bounded from below while above it covers the whole real line. Discussion of further properties such as the discrete spectrum in the subcritical case or time evolution of wave packets is postponed to a later paper. We note that the results discussed here depend substantially on the asymptotic behavior of the potential channels and would not change if the potential is modified in the vicinity of the $x$-axis, for instance, by replacing the cut-off functions in (\ref{H1}) and (\ref{HN}) with $\chi_{|y|\ge a}$ and $\chi_{|y|\ge a_j}$, respectively. It is also not important that in contrast to the original model with the Hamiltonian (\ref{HSmil}) we assume that the potential channels depth increases in both directions parallel to the $y$-axis.

%%%%%%%%%%%%%%%%%%%%%%%%%%%%%%%%%%%%%%%%%%%%%%%%%
\section{Subcritical case} \label{s: boundedness}
\setcounter{equation}{0}

To state the result we will employ a one-dimensional comparison operator
 % ------------- %
\begin{equation}
\label{comparison}
L=-\frac{\mathrm{d}^2}{\mathrm{d}x^2}+\omega^2-\lambda V(x)
\end{equation}
 % ------------- %
on $L^2(\mathbb{R})$ with the domain $H^2(\mathbb{R})$; as long as there is no danger of misunderstanding we refrain from labeling the symbol by $\omega,\,\lambda$ and $V$. The important property will be the sign of its spectral threshold; since $V$ is supposed to be nonnegative, the latter is a monotonous function of $\lambda$ and there is a $\lambda_\mathrm{crit} >0$ at which the sign changes.

We shall first focus on the subcritical coupling case.

 % ------------- %
\begin{theorem}
Under the stated assumption, the spectrum of operator $H$ given by (\ref{H1}) is bounded from below provided the operator $L$ is positive.
\end{theorem}
 % ------------- %
\begin{proof}
It is obvious it sufficient to prove the claim for $\lambda=1$. We employ Neumann bracketing. Let $h_n$ and $\widetilde{h}_n$ be the restrictions of operator $H$ to the strips $G_n=\mathbb{R} \times \left\{y:\: \ln n<y\le\ln(n+1)\right\}$ and $\widetilde{G}_n=\mathbb{R} \times \left\{y:\: -\ln(n+1)<y\le-\ln n\right\},\:n=1,2, \ldots$, with Neumann boundary conditions. Then we have the inequality
 % ------------- %
\begin{equation}\label{Neumann} H\ge\bigoplus_{n=1}^\infty h_n\oplus \widetilde{h}_n \,,
\end{equation}
 % ------------- %
and to prove the claim we have to demonstrate that the sets $\sigma(h_n)$ and $\sigma (\widetilde{h}_n)$ have a uniform lower bound as $n\to\infty$. Using the fact that the function $V$ has a bounded derivative we find
 % ------------- %
$$
V(x y)-V(x\ln n)=\mathcal{O}\left(\frac{1}{n\ln n}\right)\,, \quad
y^2-\ln^2n=\mathcal{O}\left(\frac{\ln n}{n}\right)\,,
$$
 % ------------- %
for any  $(x,y) \in G_n$, and consequently
 % ------------- %
$$ %\begin{equation}
%\label{derivative}
y^2V(xy)-\ln^2\!n\,V(x\ln n)=\mathcal{O}\left(\frac{\ln n}{n}\right)\,.
$$ %\end{equation}
 % ------------- %
Similarly, we have
 % ------------- %
$$ %\begin{equation}
%\label{derivative*}
y^2V(xy)-\ln^2\!n\,V(-x\ln n)=\mathcal{O}\left(\frac{\ln n}{n}\right).
$$ %\end{equation}
 % ------------- %
for for any $(x,y)\in\widetilde{G}_n$. These relations yield asymptotic inequalities
 % ------------- %
\begin{eqnarray}
\inf\sigma(h_n)\ge\inf\sigma(l_n)+\mathcal{O}\left(\frac{\ln n}{n}\right)\,, \nonumber \\[-.5em] \label{l_n} \\[-.5em]
\inf\sigma\left(\widetilde{h}_n\right)\ge\inf\sigma\left(\widetilde{l}_n\right)+
\mathcal{O}\left(\frac{\ln n}{n}\right)\,, \nonumber
\end{eqnarray}
 % ------------- %
in which the Neumann operators $l_n:=-\frac{\partial^2}{\partial x^2}-\frac{\partial^2}{\partial y^2}+\omega^2\ln^2\!n-\ln^2\!n\,V(x\ln n)$ on $G_n$ and $\widetilde{l}_n:=-\frac{\partial^2}{\partial x^2}-\frac{\partial^2}{\partial y^2}+\omega^2\ln^2\!n-\ln^2\!n\,V(-x\ln n)$ on $\widetilde{G}_n$ have separated variables. Since the minimal eigenvalue of $-\frac{\mathrm{d}^2}{\mathrm{d}y^2}$ on the interval
with Neumann boundary conditions defined on intervals $(\ln n <y \le\ln(n+1)),\:n=1,2,\ldots,$ is zero, we have $\inf\sigma(l_n) =\inf\sigma\big(l^{(1)}_n\big)$, where
 % ------------- %
$$
l^{(1)}_n=-\frac{\mathrm{d}^2}{\mathrm{d}x^2} +\omega^2\,\ln^2\!n-\ln^2\!n\,V(x\ln n)
$$
 % ------------- %
acts on $L^2(\mathbb{R})$. Note that the cut-off function $\chi_{\{|x|\le a\}}$ in (\ref{H1}) plays no role in the asymptotic estimate as it affects a finite number of terms only. By the change of variable $x=\frac{t}{\ln n}$ the last operator is unitarily equivalent to \mbox{$\ln^2\!n\,L$} which is positive as long as $L$ is positive. In the same way one proves that $\widetilde{l}_n$ is positive under the assumption of the theorem; this in combination with (\ref{l_n}) concludes the proof.
\end{proof}

\noindent By a straightforward modification of the proof we get the following claim.

 % ------------- %
\begin{corollary} \label{c: interval}
Let $H$ be the operator on $(-c,c)\times \mathbb{R}$ for some $c\ge a$ given the differential expression (\ref{H1}) with Dirichlet (Neumann, periodic) boundary conditions in the variable $x$. The spectrum of $H$ is bounded from below if \mbox{$L\ge 0$} holds, where $L$ is the operator (\ref{comparison}) on $L^2(-c,c)$ with Dirichlet (respectively, Neumann or periodic) boundary conditions.
\end{corollary}
 % ------------- %

%%%%%%%%%%%%%%%%%%%%%%%%%%%%%%%%%%%%%%%%%%%%%%%%%%%%%
\section{Supercritical case} \label{s: unboundedness}
\setcounter{equation}{0}

Let us turn to the case when the `escape to infinity' is possible.
 % ------------- %
\begin{theorem}
Under our hypotheses, $\sigma(H)=\mathbb{R}$ holds if $\inf\sigma(L)<0$.
\end{theorem}
 % ------------- %
\begin{proof}
To prove that any real number $\mu$ belongs to essential spectrum of operator $H$ we are going to use Weyl's criterion \cite[Thm.~VII.12]{RS}: we have to find a sequence $\{\psi_k\}_{k=1}^\infty\subset D(H)$ such that $\|\psi_k\|=1$ which contains no convergent subsequence and
 % ------------- %
$$ %\begin{equation}
%\label{Weyl}
\|H\psi_k-\mu\psi_k\|\to0,\quad\text{as}\quad k\to\infty
$$ %\end{equation}
 % ------------- %
holds. Since the claim is invariant under scaling transformations we can suppose without loss of generality that $\inf\sigma(L)=-1$. The spectral threshold is easily seen to be a simple isolated eigenvalue; we denote the corresponding normalized eigenfunction of $L$ by $h$. Our aim is to show first that $0\in\sigma_{\mathrm{ess}}(H)$.

We fix a positive $\varepsilon$ and choose a natural number $k=k(\varepsilon)$ with which we associate a function $\chi_k\subset C_0^2(1,k)$ satisfying the following conditions
 % ------------- %
\begin{equation}\label{chi.conditions}
\int_1^k\frac{1}{z}\chi_k^2(z)\,\mathrm{d}z=1
\quad\text{and}\quad
\int_1^k z(\chi'_k(z))^2\,\mathrm{d}z<\varepsilon.
\end{equation}
 % ------------- %
To give an example, consider the function
 % ------------- %
$$
\tilde{\chi}_k(z)=\frac{8\ln^3z}{\ln^3k}\chi_{\left\{1\le z\le\sqrt{k}\right\}}(z)
+\frac{2\ln k-2\ln z}{\ln k}\chi_{\left\{\sqrt{k}+1\le z\le k-1\right\}} (z)
$$$$+g_k(z)\chi_{\left\{\sqrt{k}<z<\sqrt{k}+1\right\}}(z)+q_k(z)\chi_{\left\{k-1<z\le k\right\}} (z),
$$
 % ------------- %
where $g_k$ and $q_k$ are interpolating functions chosen in such a way that $\tilde{\chi}_k\in C_0^2(1,k)$. The first integral in (\ref{chi.conditions}) is positive for $\tilde{\chi}_k$, in fact we have $\int_1^{\sqrt{k}} \frac{1}{z}\tilde{\chi}_k^2(z)\, \mathrm{d}z\ge \frac14$, hence we can define $\chi_k(z)=\left(\int_1^k\frac{1}{z} \tilde{\chi}_k^2(z)\, \mathrm{d}z\right)^{-1/2}\tilde{\chi}_k(z)$. This function satisfies by definition the first condition of (\ref{chi.conditions}) and one can check that it also satisfies the second one provided $k$ is sufficiently large; this follows from the fact that $\int_1^k z(\chi'_k(z))^2\,\mathrm{d}z=\mathcal{O}\left(\frac{1}{\ln k}\right)$ as $k\to\infty$.

Such functions allow us to construct the Weyl sequence we seek. Given a function $\chi_k$ with the described properties, we define
 % ------------- %
\begin{equation}
\label{sequence}
\psi_k(x,y):=h(x y)\,\e^{iy^2/2}\chi_k\left(\frac{y}{n_k}\right)
+\frac{f(x y)}{y^2}\, \e^{iy^2/2}\chi_k\left(\frac{y}{n_k}\right)\,,
\end{equation}
 % ------------- %
where $f(t):=-\frac{i}{2}\,t^2h(t),\:t\in\mathbb{R}$, and $n_k\in\mathbb{N}$ is a positive integer to chosen later. For the moment we just note that choosing $n_k$ large enough for a given $k$ one can achieve that $\|\psi_k\|_{L^2(\mathbb{R}^2)}\ge\frac{1}{2}$ as the following estimates show,
 % ------------- %
\begin{eqnarray}
\lefteqn{\int_{\mathbb{R}^2}\left|h(x y)\,\e^{iy^2/2}\chi_k\left(\frac{y}{n_k}\right)
\right|^2\,\mathrm{d}x\,\mathrm{d}y=\int_{n_k}^{kn_n}\int_{\mathbb{R}}
\left|h(x y)\chi_k\left(\frac{y}{n_k}\right)
\right|^2\,\mathrm{d}x\,\mathrm{d}y} \nonumber \\ &&
=\int_{n_k}^{kn_n}\int_{\mathbb{R}}
\frac{1}{y}\left|h(t)\chi_k\left(\frac{y}{n_k}\right)
\right|^2\,\mathrm{d}t\,\mathrm{d}y=\int_{\mathbb{R}}|h(t)|^2\,\mathrm{d}t
\,\int_{n_k}^{kn_n}\frac{1}{y}\left|\chi_k\left(\frac{y}{n_k}\right)
\right|^2\,\mathrm{d}y \nonumber \\ &&
\label{firstpart}
=\int_{n_k}^{kn_n}\frac{1}{y}\left|\chi_k\left(\frac{y}{n_k}\right)
\right|^2\,\mathrm{d}y=\int_1^k\frac{1}{z}\left|\chi_k(z)
\right|^2\,\mathrm{d}z=1
\end{eqnarray}
 % ------------- %
and
 % ------------- %
\begin{eqnarray}
\lefteqn{\int_{\mathbb{R}^2}\left|\frac{1}{y^2}
f(x y)\,\e^{iy^2/2}\chi_k\left(\frac{y}{n_k}\right)
\right|^2\,\mathrm{d}x\,\mathrm{d}y=\int_{n_k}^{kn_k}\int_{\mathbb{R}}
\left|\frac{1}{y^2}f(x y)\,\chi_k\left(\frac{y}{n_k}\right)
\right|^2\,\mathrm{d}x\,\mathrm{d}y} \nonumber \\ && =\int_{n_k}^{kn_k}\int_{\mathbb{R}}
\frac{1}{y}\left|\frac{1}{y^2}f(t)\chi_k\left(\frac{y}{n_k}\right)\right|^2\,
\mathrm{d}t\,\mathrm{d}y\le\frac{1}{n_k^5}\int_{n_k}^{kn_k}\int_{\mathbb{R}}
\left|f(t)\chi_k\left(\frac{y}{n_k}\right)\right|^2\,
\mathrm{d}t\,\mathrm{d}y \nonumber \\ && =\frac{1}{n_k^4}\int_{\mathbb{R}}|f(t)|^2\,\mathrm{d}t
\,\int_1^k|\chi_k(z)|^2\,\mathrm{d}z<\frac{1}{16}\,; \label{secondpart}
\end{eqnarray}
 % ------------- %
note that since the potential $V$ has a compact support by assumption, the ground state eigenfunction $h$ decays exponentially as $|x|\to\infty$, hence the first integral in the last expression converges.

Our next aim is to show that $\|H\psi_k\|_{L^2(\mathbb{R}^2)}^2<c\varepsilon$ with a fixed $c$ holds for $k=k(\varepsilon)$. By a straightforward calculation one gets
 % ------------- %
$$
\frac{\partial^2\psi_k}{\partial x^2}=y^2h''(x y)\,\e^{iy^2/2}\chi_k\left(\frac{y}{n_k}
\right)+f''(x y)\,\e^{iy^2/2}\chi_k\left(\frac{y}{n_k}\right)
$$
 % ------------- %
and
 % ------------- %
\begin{eqnarray}
\lefteqn{\frac{\partial^2\psi_k}{\partial y^2}
=x^2h''(x y)\,\e^{iy^2/2}\chi_k\left(\frac{y}{n_k}\right)
+2ix y h'(x y)\,\e^{iy^2/2}\chi_k\left(\frac{y}{n_k}\right)} \nonumber \\ &&
+\frac{2x}{n_k}h'(x y)\,\e^{iy^2/2}\chi_k'\left(\frac{y}{n_k}\right)
-y^2h(x y)\,\e^{iy^2/2}\chi_k\left(\frac{y}{n_k}\right) \nonumber \\ &&
+i h(x y)e^{iy^2/2}\chi_k\left(\frac{y}{n_k}\right)
+2\frac{i y}{n_k}h(x y)\,\e^{iy^2/2}\chi'_k\left(\frac{y}{n_k}\right)\nonumber \\ &&
+\frac{1}{n_k^2}h(x y)\,\e^{iy^2/2}\chi_k''\left(\frac{y}{n_k}\right)
+\frac{x^2}{y^2}f''(x y)\,\e^{iy^2/2}\chi_k\left(\frac{y}{n_k}\right)\nonumber \\ &&
+2\frac{i x}{y}f'(x y)\,\e^{iy^2/2}\chi_k\left(\frac{y}{n_k}\right)
+\frac{2x}{n_k y^2}f'(x y)\,\e^{iy^2/2}\chi_k'\left(\frac{y}{n_k}\right) \nonumber \\ &&
-f(x y)\,\e^{iy^2/2}\chi_k\left(\frac{y}{n_k}\right)
+\frac{i}{y^2}f(x y)\,\e^{iy^2/2}\chi_k\left(\frac{y}{n_k}\right) \nonumber \\ &&
+\frac{1}{y^2 n_k^2}f(x y)\,\e^{iy^2/2}\chi_k''\left(\frac{y}{n_k}\right)
+\frac{2i}{n_ky}f(x y)\,\e^{iy^2/2}\chi_k'\left(\frac{y}{n_k}\right)\nonumber \\ &&
-\frac{4x}{y^3}f'(x y)\,\e^{iy^2/2}\chi_k\left(\frac{y}{n_k}\right)
-\frac{4i}{y^2}f(x y)\,\e^{iy^2/2}\chi_k\left(\frac{y}{n_k}\right) \nonumber \\ &&
-\frac{4}{n_k y^3}f(x y)\,\e^{iy^2/2}\chi_k'\left(\frac{y}{n_k}\right)
+\frac{6}{y^4}f(x y)e^{iy^2/2}\chi_k\left(\frac{y}{n_k}\right). \label{calculations}
\end{eqnarray}
 % ------------- %
We want to show that choosing $n_k$ sufficiently large one can make the terms at the right hand side of (\ref{calculations}) as small as we wish. Changing the integration variables, we get the following estimate,
 % ------------- %
\begin{eqnarray*}
\lefteqn{\int_{\mathbb{R}^2}\left|x^2\,h''(x y)\,\e^{iy^2/2}\chi_k\left(\frac{y}{n_k}
\right)\right|^2\,\mathrm{d}x\,\mathrm{d}y=\int_{n_k}^{kn_{k}}\int_{\mathbb{R}}
\left|x^2\,h''(x y) \chi_k\left(\frac{y}{n_k}\right)
\right|^2\,\mathrm{d}x\,\mathrm{d}y} \\ &&
\hspace{-2em} =\int_{n_k}^{kn_k}\frac{1}{y^5}
\left|\chi_k\left(\frac{y}{n_k}\right)\right|^2\,\mathrm{d}y
\,\int_{\mathbb{R}}t^4|h''(t)|^2\,\mathrm{d}t \le\frac{1}{n_k^4}
\int_1^k|\chi_k(z)|^2\mathrm{d}z\,\int_{\mathbb{R}}t^4|h''(t)|^2\,\mathrm{d}t\,,
\end{eqnarray*}
 % ------------- %
where the last integral again converges from the reason described above. In the same way we establish the remaining inequalities which we need to demonstrate our claim:
 % ------------- %
\begin{eqnarray*}
\lefteqn{\int_{\mathbb{R}^2}\left|\frac{x}{n_k}h'(x y)\,\e^{iy^2/2}\chi_k'\left(\frac{y}{n_k}\right)
\right|^2\,\mathrm{d}x\,\mathrm{d}y=\int_{n_k}^{kn_k}\int_{\mathbb{R}}
\frac{1}{y}\left|\frac{t}{n_k y}h'(t)\chi_k'\left(\frac{y}{n_k}\right)\right|^2
\,\mathrm{d}t\,\mathrm{d}y} \\ && \le\frac{1}{n_k^4}
\int_1^k|\chi_k'(z)|^2\mathrm{d}z\,\int_{\mathbb{R}}t^2|h'(t)|^2\,\mathrm{d}t\,,
\phantom{AAAAAAAAAAAAAAAAAAAAAAAA}
\end{eqnarray*}
 % ------------- %
\begin{eqnarray*}
\lefteqn{\int_{\mathbb{R}^2}\left|\frac{1}{n_k^2}h(x y)\,\e^{iy^2/2}\chi_k''\left(\frac{y}{n_k}
\right)\right|^2\,\mathrm{d}x\,\mathrm{d}y=\frac{1}{n_k^4}\int_{n_k}^{kn_k}
\int_{\mathbb{R}}\frac{1}{y}\left|h(t)\chi_k''\left(\frac{y}{n_k}\right)
\right|^2\,\mathrm{d}t\,\mathrm{d}y} \\ && \le\frac{1}{n_k^4}\int_1^k|\chi''_k(z)|^2\mathrm{d}z
\,\int_{\mathbb{R}}|h(t)|^2\,\mathrm{d}t\,,
\phantom{AAAAAAAAAAAAAAAAAAAAAAAA}
\end{eqnarray*}
 % ------------- %
\begin{eqnarray*}
\lefteqn{\int_{\mathbb{R}^2}\left|\frac{x^2}{y^2}f''(x y)\,\e^{iy^2/2}\chi_k\left(\frac{y}{n_k}
\right)\right|^2\,\mathrm{d}x\,\mathrm{d}y=\int_{n_k}^{kn_k}\int_{\mathbb{R}}
\frac{1}{y}\left|\frac{t^2}{y^4}f''(t)\chi_k\left(\frac{y}{n_k}\right)\right|^2\,
\mathrm{d}t\,\mathrm{d}y} \\ && \le\frac{1}{n_k^8}\int_1^k|\chi_k(z)|^2
\mathrm{d}z\,\int_{\mathbb{R}}t^4|f''(t)|^2\,\mathrm{d}t\,,
\phantom{AAAAAAAAAAAAAAAAAAAAAAAA}
\end{eqnarray*}
 % ------------- %
\begin{eqnarray*}
\lefteqn{\int_{\mathbb{R}^2}\left|\frac{x}{y}f'(x y)\,\e^{iy^2/2}\chi_k\left(\frac{y}{n_k}\right)
\right|^2\,\mathrm{d}x\,\mathrm{d}y
=\int_{n_k}^{kn}\int_{\mathbb{R}}\frac{1}{y}\left|\frac{t}{y^2}f'(t)
\chi_k\left(\frac{y}{n_k}\right)\right|^2\,\mathrm{d}t\,\mathrm{d}y} \\ &&
\le\frac{1}{n_k^4}\int_1^k|\chi_k(z)|^2
\mathrm{d}z\,\int_{\mathbb{R}}t^2|f'(t)|^2\,\mathrm{d}t\,,
\phantom{AAAAAAAAAAAAAAAAAAAAAAAA}
\end{eqnarray*}
 % ------------- %
\begin{eqnarray*}
\lefteqn{\int_{\mathbb{R}^2}\left|\frac{x}{n_k y^2}f'(x y)\,\e^{iy^2/2}\chi_k'\left(\frac{y}{n_k}
\right)\right|^2\,\mathrm{d}x\,\mathrm{d}y
=\int_{n_k}^{kn_k}\int_{\mathbb{R}}\frac{1}{y}\left|\frac{t}{n_k y^3}f'(t)
\chi_k'\left(\frac{y}{n_k}\right)\right|^2\,\mathrm{d}t\,\mathrm{d}y} \\ &&
\le\frac{1}{n_k^8}\int_1^k|\chi_k'(z)|^2
\mathrm{d}z\,\int_{\mathbb{R}}t^2|f'(t)|^2\,\mathrm{d}t\,,
\phantom{AAAAAAAAAAAAAAAAAAAAAAAA}
\end{eqnarray*}
 % ------------- %
\begin{eqnarray*}
\lefteqn{\int_{\mathbb{R}^2}\left|\frac{1}{y^2}f(x y)\,\e^{iy^2/2}\chi_k\left(\frac{y}{n_k}\right)
\right|^2\,\mathrm{d}x\,\mathrm{d}y=\int_{n_k}^{kn_k}
\int_{\mathbb{R}}\frac{1}{y}\left|\frac{1}{y^2}f(t)\chi_k\left(\frac{y}{n_k}\right)
\right|^2\,\mathrm{d}t\,\mathrm{d}y} \\ &&
\le\frac{1}{n_k^4}\int_1^k|\chi_k(z)|^2
\mathrm{d}z\,\int_{\mathbb{R}}|f(t)|^2\,\mathrm{d}t\,,
\phantom{AAAAAAAAAAAAAAAAAAAAAAAA}
\end{eqnarray*}
 % ------------- %
\begin{eqnarray*}
\lefteqn{\int_{\mathbb{R}^2}\left|\frac{1}{n_k^2y^2}
f(x y)\,\e^{iy^2/2}\chi_k''\left(\frac{y}{n_k}\right)
\right|^2\,\mathrm{d}x\,\mathrm{d}y=\int_{n_k}^{kn_k}\int_{\mathbb{R}}
\frac{1}{y}\left|\frac{1}{n_k^2y^2}f(t)\chi_k''\left(\frac{y}{n_k}\right)\right|^2\,
\mathrm{d}t\,\mathrm{d}y} \\ &&
\le\frac{1}{n_k^8}\int_1^k|\chi_k''(z)|^2
\mathrm{d}z\,\int_{\mathbb{R}}|f(t)|^2\,\mathrm{d}t\,,
\phantom{AAAAAAAAAAAAAAAAAAAAAAAA}
\end{eqnarray*}
 % ------------- %
\begin{eqnarray*}
\lefteqn{\int_{\mathbb{R}^2}\left|\frac{1}{n_k y}f(x y)\,\e^{iy^2/2}\chi'_k\left(\frac{y}{n_k}
\right)\right|^2\,\mathrm{d}x\,\mathrm{d}y=\int_{n_k}^{kn_k}
\int_{\mathbb{R}}\frac{1}{y}\left|\frac{1}{n_k y}f(t)\chi'_k\left(\frac{y}{n_k}\right)
\right|^2\,\mathrm{d}t\,\mathrm{d}y} \\ && \le\frac{1}{n_k^4}\int_1^k|\chi'_k(z)|^2
\mathrm{d}z\,\int_{\mathbb{R}}|f(t)|^2\,\mathrm{d}t\,,
\phantom{AAAAAAAAAAAAAAAAAAAAAAAA}
\end{eqnarray*}
 % ------------- %
\begin{eqnarray*}
\lefteqn{\int_{\mathbb{R}^2}\left|\frac{x}{y^3}f'(x y)\,\e^{iy^2/2}\chi_k\left(\frac{y}{n_k}\right)
\right|^2\,\mathrm{d}x\,\mathrm{d}y=\int_{n_k}^{kn_k}\int_{\mathbb{R}}
\frac{1}{y}\left|\frac{t}{y^4}f'(t)\chi_k\left(\frac{y}{n_k}\right)\right|^2\,
\mathrm{d}t\,\mathrm{d}y} \\ &&
\le\frac{1}{n_k^8}\int_1^k|\chi_k(z)|^2\,\mathrm{d}z\,\int_{\mathbb{R}}t^2|f'(t)|^2\,\mathrm{d}t\,,
\phantom{AAAAAAAAAAAAAAAAAAAAAAAA}
\end{eqnarray*}
 % ------------- %
\begin{eqnarray*}
\lefteqn{\int_{\mathbb{R}^2}\left|\frac{1}{y^4}f(x y)\,\e^{iy^2/2}\chi_k\left(\frac{y}{n_k}\right)
\right|^2\,\mathrm{d}x\,\mathrm{d}y=
\int_{n_k}^{kn_k}\int_{\mathbb{R}}\frac{1}{y}\left|\frac{1}{y^4}f(t)
e^{iy^2/2}\chi_k\left(\frac{y}{n_k}\right)\right|^2\,\mathrm{d}t\,\mathrm{d}y} \\ &&
\le\frac{1}{n_k^8}\int_1^k|\chi_k(z)|^2
\mathrm{d}z\,\int_{\mathbb{R}}|f(t)|^2\,\mathrm{d}t\,.
\phantom{AAAAAAAAAAAAAAAAAAAAAAAA}
\end{eqnarray*}
 % ------------- %
Consequently, choosing $n_k$ large enough we can achieve that the sum of all the integrals at the left-hand sides of the above inequalities is less than $\varepsilon$. Then we have
 % ------------- %
\begin{eqnarray*}
&& \int_{\mathbb{R}^2}|H\psi_k|^2(x,y)\,\mathrm{d}x\,\mathrm{d}y<\int_{n_k}^{kn_k}
\int_{\mathbb{R}}\biggl|y^2h''(xy)\chi_k\left(\frac{y}
{n_k}\right)+f''(xy)\chi_k\left(\frac{y}{n_k}\right) \\ &&
\qquad +2ixyh'(xy)\chi_k\left(\frac{y}
{n_k}\right)+ih(xy)\chi_k\left(\frac{y}{n_k}
\right) -y^2h(xy)\chi_k\left(\frac{y}
{n_k}\right) \\ && \qquad
+\frac{2iy}{n_k}h(xy)\chi_k'\left(\frac{y}{n_k}\right)-f(xy)
\chi_k\left(\frac{y}{n_k}\right) -\omega^2\,y^2h(xy)\chi_k\left(\frac{y}{n_k}\right)
\\ && \qquad -\omega^2\,f(xy)\chi_k\left(\frac{y}{n_k}\right) +y^2\,V(xy)h(xy)\chi_k\left(\frac{y}{n_k}\right)
\\ && \qquad +V(xy)f(xy)\chi_k\left(\frac{y}{n_k}\right)\biggr|^2\,\mathrm{d}x\,\mathrm{d}y+
\varepsilon \\ &&
=\int_{n_k}^{kn_k}\int_{\mathbb{R}}\biggl|y^2\left(h''(xy)-\omega^2h(xy)+V(xy)h(xy)-
h(xy)\right)\chi_k\left(\frac{y}{n_k}\right) \\ && \qquad
+ih(xy)\chi_k\left(\frac{y}{n_k}\right)
+f''(xy)\chi_k\left(\frac{y}{n_k}\right)+2ixyh'(xy)\chi_k\left(\frac{y}{n_k}\right) \\ && \qquad +\frac{2iy}{n_k}h(xy)\chi'_k\left(\frac{y}{n_k}\right)-f(xy)\chi_k\left(\frac{y}{n_k}
\right) -\omega^2\,f(xy)\chi_k\left(\frac{y}{n_k}\right) \\ && \qquad
+V(xy)f(xy)\chi_k\left(\frac{y}{n_k}\right)\biggr|^2\,\mathrm{d}x\,\mathrm{d}y
+\varepsilon\,.
\end{eqnarray*}
 % ------------- %
Using the fact that $Lh=-h$ and applying the Cauchy inequality, the above result implies
 % ------------- %
\begin{eqnarray*}
&& \int_{\mathbb{R}^2}|H\psi_k|^2(x,y)\,\mathrm{d}x\,\mathrm{d}y<\int_{n_k}^{kn_k}
\int_{\mathbb{R}}\biggl|\biggl(f''(xy)+2ixyh'(xy)+ih(xy)-f(xy) \\ && \qquad
-\omega^2\,f(xy) +V(xy)f(xy)\biggr)\chi_k\left(\frac{y}{n_k}\right)+\frac{2iy}{n_k}h(xy)\chi_k'
\left(\frac{y}{n_k}\right)\biggr|^2\,\mathrm{d}x\,\mathrm{d}y
+\varepsilon \\ && \le2\int_1^k\frac{1}{z}|\chi_k(z)|^2\,\mathrm{d}z\,\int_{\mathbb{R}}
\biggl|-f''(t)+f(t)\left(1+\omega^2-V(t)\right)-2ith'(t)-ih(t))\biggr|^2
\mathrm{d}t \\ && \qquad
+8\int_1^kz|\chi_k'(z)|^2\,\mathrm{d}z+\varepsilon \\ && \le2\int_{\mathbb{R}}
\biggl|-f''(t)+f(t)\left(1+\omega^2-V(t)\right)-2ith'(t)-ih(t)\biggr|^2
\,\mathrm{d}t+9\varepsilon\,.
\end{eqnarray*}
 % ------------- %
It is easy to check that
$-(t^2h(t))''+t^2h(t)(1+\omega^2-V(t))=-4th'(t)-2h(t)$, hence $-f''(t)+f(t)\left(1+\omega^2-V(t)\right)-2ith'(t)-ih(t)=0$ and the last integral in the above estimate vanishes, which gives
 % ------------- %
\begin{equation}
\label{final}
\int_{\mathbb{R}^2}|H\psi_k|^2(x,y)\,\mathrm{d}x\,\mathrm{d}y\le 9\varepsilon\,.
\end{equation}
 % ------------- %
To complete the proof we fix a sequence $\{\varepsilon_j\}_{j=1}^\infty$ such that $\varepsilon_j\searrow0$ holds as $j\to\infty$ and to any $j$ we construct a function $\psi_{k(\varepsilon_j)}$ with the corresponding numbers chosen in such a way that $n_{k(\varepsilon_j)} >k(\varepsilon_{j-1}) n_{k(\varepsilon_{j-1})}$. The norms of $H\psi_{k(\varepsilon_j)}$ satisfy inequality which (\ref{final}) with $9\varepsilon_j$ on the right-hand side, and since the supports of $\psi_{k(\varepsilon_j)},\:j =1,2,\ldots,$ do not intersect each other by construction, their sequence converges weakly to zero. This yields the sought Weyl sequence for zero energy; for any nonzero real number $\mu$ we use the same procedure replacing the above $\psi_k$ with
 % ------------- %
$$
\psi_k(x,y)=h(xy)\,\e^{i\epsilon_\mu(y)} \chi_k\left(\frac{y}{n_k}\right)+
\frac{f(xy)}{y^2}\,\e^{i\epsilon_\mu(y)} \chi_k\left(\frac{y}{n_k}\right)\,,
$$
 % ------------- %
where $\epsilon_\mu(y):= \displaystyle{\int_{\sqrt{|\mu|}}^y\sqrt{t^2+\mu}\,\mathrm{d}t}$, and furthermore, the functions $f,\,\chi_k$ defined in the same way as above.
\end{proof}

%%%%%%%%%%%%%%%%%%%%%%%%%%%%%%%%%%%%%%%%%%%%%%%%%%%%%%%%%%%%%%%%%%%%%%%%%%%%%%
\section{Intervals and multiple channels}
\setcounter{equation}{0}\label{s: restr-mult}

Let us look next how the above result changes if the motion in the $x$ direction is restricted. We have the following result:

 % ------------- %
\begin{theorem} \label{th: interval}
Let $H$ be the operator on $L^2(-c,c)\otimes L^2(\R)$ for some $c>0$ given by the differential expression (\ref{H1}) with Dirichlet condition at $x=\pm c$ and denote by $L$ the corresponding Dirichlet operator (\ref{comparison}) on $L^2(-c,c)$. If the spectral threshold of $L$ is negative, the spectrum of $H$ covers the whole real axis.
\end{theorem}
 % ------------- %
\begin{proof}
Without loss of generality we may suppose that $c=1$. We shall apply again Weyl's criterion modifying the argument of the previous section.
By Dirichlet bracketing, one has that $\widetilde{L}\le\oplus_{k=1}^3 L_k$, where $\widetilde{L}$ is the original comparison operator (\ref{comparison}), $\widetilde{L}=-\frac{\mathrm{d}^2}{\mathrm{d}x^2}+\omega^2-V$ on $L^2(\mathbb{R})$, while $L_1$ and $L_j,\:j=2,3,$ are given by the same differential expression on $L^2(-1,1)$ and $L^2(-\infty,-1),\, L^2(1,\infty)$, respectively. Thus under the assumption the spectral threshold of $\widetilde{L}$ is negative, and without loss of generality we may suppose that its ground state satisfies $\widetilde{L}h=-h$ with $\|h\|=1$ and show that $0\in\sigma_{\mathrm{ess}}(H)$. The functions (\ref{sequence}) are now changed as follows,
 % ------------- %
$$
\psi_k(x,y)=h(x y)\,\e^{iy^2/2}\chi_k\left(\frac{y}{n_k}\right)\phi(x)
+\frac{f(x y)}{y^2}\,\e^{iy^2/2}\chi_k\left(\frac{y}{n_k}\right)\phi(x)
$$
 % ------------- %
with $\phi\in C_0^2(-1,1)$ such that $\phi(x)=1$ holds for $|x|\le\frac{1}{2}$, while the numbers $k=k(\varepsilon),\,n_k\in\mathbb{N}$ and functions $\chi_k,\, f$ are the same as before. Instead of the estimates (\ref{firstpart}) and (\ref{secondpart}) we now have for large enough $n_k$ the inequalities
 % ------------- %
\begin{eqnarray*}
\lefteqn{\int_{-1}^1\int_{n_k}^{kn_k}\left|h(xy)\,\e^{iy^2/2}\chi_k\left(\frac{y}{n_k}\right)
\phi(x)\right|^2\,\mathrm{d}x\,\mathrm{d}y} \\ &&
\ge\int_{n_k}^{kn_k}\int_{-1/2}^{1/2}\left|h(xy)\chi_k\left(\frac{y}{n_k}\right)
\right|^2\,\mathrm{d}t\,\mathrm{d}y \\ &&
=\int_{n_k}^{kn_k}\int_{-y/2}^{y/2}\frac{1}{y}\left|h(t)\chi_k\left(\frac{y}{n_k}\right)
\right|^2\,\mathrm{d}t\,\mathrm{d}y \\ && \ge\int_{n_k}^{kn_k}\int_{-n_k/2}^{n_k/2}\frac{1}{y}\left|h(t)
\chi_k\left(\frac{y}{n_k}\right)\right|^2\,\mathrm{d}t\,\mathrm{d}y \\ &&
= \int_1^k\frac{1}{z}|\chi_k(z)|^2\,\mathrm{d}z\,\int_{-n_k/2}^{n_k/2}|h(t)|^2\,\mathrm{d}t
\ge\frac{1}{2}
\end{eqnarray*}
 % ------------- %
and
 % ------------- %
\begin{eqnarray*}
\lefteqn{\int_{-1}^1\int_{n_k}^{kn_k}\left|\frac{f(x y)}{y^2}\,\e^{iy^2/2}\chi_k\left(\frac{y}{n_k}\right)\phi(x)
\right|^2\,\mathrm{d}x\,\mathrm{d}y} \\ &&
\le\|\phi\|_{L^\infty(\mathbb{R})}^2
\int_{-1}^1\int_{n_k}^{kn_k}\left|\frac{f(xy)}{y^2}\chi_k\left(\frac{y}{n_k}\right)
\right|^2 \mathrm{d}x\,\mathrm{d}y \\ &&
\le\|\phi\|_{L^\infty(\mathbb{R})}^2\int_{\mathbb{R}}\int_{n_k}^{kn_k}
\left|\frac{f(x y)}{y^2}\chi_k\left(\frac{y}{n_k}\right)\right|^2\,\mathrm{d}x\,\mathrm{d}y \\ &&
=\|\phi\|_{L^\infty(\mathbb{R})}^2\int_{n_k}^{kn_k}\int_{\mathbb{R}}\frac{1}{y}
\left|\frac{f(t)}{y^2}\chi_k\left(\frac{y}{n_k}\right)\right|^2\,\mathrm{d}t\,\mathrm{d}y \\ &&
\le\frac{1}{n_k^4}\|\phi\|_{L^\infty(\mathbb{R})}^2\int_1^k|\chi_k(z)|^2\,\mathrm{d}z\,\int_{\mathbb{R}}
|f(t)|^2\,\mathrm{d}t<\varepsilon\,,
\end{eqnarray*}
 % ------------- %
which means that $\|\psi_k\|_{L^2(\mathbb{R}^2)}\ge\frac{1}{2}-2\sqrt{\varepsilon}$ holds for $n_k$ large enough; our aim is to show that $\|H\psi_k\|_{L^2(\mathbb{R}^2)}^2<d\varepsilon$ with a fixed $d>0$. Let us first compute the partial derivatives
 % ------------- %
\begin{eqnarray*}
\lefteqn{\frac{\partial^2\psi_k}{\partial x^2}=y^2h''(x y)\,\e^{iy^2/2}\chi_k\left(\frac{y}{n_k}
\right)\phi(x)+2yh'(xy)\,\e^{iy^2/2}\chi_k\left(\frac{y}{n_k}\right)\phi'(x)} \\ && +
h(xy)\,\e^{iy^2/2}\chi_k\left(\frac{y}{n_k}\right)\phi''(x)+f''(x y)\,\e^{iy^2/2}\chi_k
\left(\frac{y}{n_k}\right)\phi(x) \\ && +\frac{2}{y}f'(xy)\,\e^{iy^2/2}\chi_k\left(\frac{y}{n_k}\right)\phi'(x)+
\frac{1}{y^2}f(x y)\,\e^{iy^2/2}\chi_k\left(\frac{y}{n_k}\right)\phi''(x)
\end{eqnarray*}
 % ------------- %
and
 % ------------- %
\begin{eqnarray*}
\lefteqn{\frac{\partial^2\psi_k}{\partial y^2}=x^2h''(x y)\,\e^{iy^2/2}\chi_k\left(\frac{y}{n_k}
\right)\phi(x)+2ix y h'(x y)\,\e^{iy^2/2}\chi_k\left(\frac{y}{n_k}\right)\phi(x)} \\ &&
+\frac{2x}{n_k}h'(x y)\,\e^{iy^2/2}\chi_k'\left(\frac{y}{n_k}\right)\phi(x)-y^2h(x y)\,\e^{iy^2/2}
\chi_k\left(\frac{y}{n_k}\right)\phi(x) \\ && +i h(x y)\,\e^{iy^2/2}
\chi_k\left(\frac{y}{n_k}\right)\phi(x)+2\frac{i y}{n_k}h(x y)\,\e^{iy^2/2}
\chi'_k\left(\frac{y}{n_k}\right)\phi(x) \\ &&
+\frac{1}{n_k^2}h(x y)\,\e^{iy^2/2}\chi_k''\left(\frac{y}{n_k}\right)\phi(x)
+\frac{x^2}{y^2}f''(x y)\,\e^{iy^2/2}\chi_k\left(\frac{y}{n_k}\right)\phi(x) \\ &&
+2\frac{i x}{y}f'(x y)\,\e^{iy^2/2}\chi_k\left(\frac{y}{n_k}\right)\phi(x)+
\frac{2x}{n_k y^2}f'(x y)\,\e^{iy^2/2}\chi_k'\left(\frac{y}{n_k}\right)\phi(x) \\ && -
f(x y)\,\e^{iy^2/2}\chi_k\left(\frac{y}{n_k}\right)\phi(x)+\frac{i}{y^2}f(x y)\,\e^{iy^2/2}
\chi_k\left(\frac{y}{n_k}\right)\phi(x) \\ && +\frac{1}{y^2 n_k^2}f(x y)
e^{iy^2/2}\chi_k''\left(\frac{y}{n_k}\right)\phi(x)+\frac{2i}{n_ky}f(x y)\,\e^{iy^2/2}\chi_k'
\left(\frac{y}{n_k}\right)\phi(x) \\ && -\frac{4x}{y^3}f'(x y)\,\e^{iy^2/2}
\chi_k\left(\frac{y}{n_k}\right)\phi(x)-\frac{4i}{y^2}f(x y)\,\e^{iy^2/2}
\chi_k\left(\frac{y}{n_k}\right)\phi(x) \\ &&
-\frac{4}{n_k y^3}f(x y)\,\e^{iy^2/2}
\chi_k'\left(\frac{y}{n_k}\right)\phi(x)+\frac{6}{y^4}f(x y)\,\e^{iy^2/2}\chi_k\left(\frac{y}{n_k}
\right)\phi(x).
\end{eqnarray*}
 % ------------- %
Using the exponential decay of $h$ and the fact that $\phi$ is constant on $[-1/2,1/2]$ we find that for all sufficiently large $n_k$ we have
 % ------------- %
\begin{eqnarray*}
\lefteqn{\int_{-1}^1\int_{n_k}^{kn_k}\left|y\,h'(x y)\,\e^{iy^2/2}\chi_k\left(\frac{y}{n_k}
\right)\phi'(x)\right|^2\,\mathrm{d}x\,\mathrm{d}y}
\\ && \le\|\phi'\|^2_{L^\infty(\mathbb{R})}\int_{|x|>\frac{1}{2}}\int_{n_k}^{kn_k}
\left|y\,h'(xy)\chi_k
\left(\frac{y}{n_k}\right)\right|^2\mathrm{d}t\,\mathrm{d}y
\\ && \le\|\phi'\|^2_{L^\infty(\mathbb{R})}\int_{n_k}^{kn_{k}}
\int_{n_k/2}^\infty\frac{1}{y}\left|y\,h'(t)\chi_k\left(\frac{y}{n_k}\right)
\right|^2\,\mathrm{d}t\,\mathrm{d}y \\ && +\|\phi'\|^2_{L^\infty(\mathbb{R})}
\int_{n_k}^{kn_{k}}\int_{-\infty}^{-n_k/2}\frac{1}{y}\left|y\,h'(t)\chi_k\left(\frac{y}
{n_k}\right)\right|^2\,\mathrm{d}t\,\mathrm{d}y \\ &&
=n_k^2\|\phi'\|^2_{L^\infty(\mathbb{R})}\int_1^k z\chi_k^2(z)\,\mathrm{d}z\,
\int_{n_k/2}^\infty|h'(t)|^2\,\mathrm{d}t \\ && +n_k^2\|\phi'\|^2_{L^\infty
(\mathbb{R})}\int_1^k z\chi_k^2(z)\,\mathrm{d}z\,\int_{-\infty}^{-n_k/2}|h'(t)|^2\,\mathrm{d}t<
\varepsilon\,.
\end{eqnarray*}
 % ------------- %
in a similar way,
 % ------------- %
\begin{eqnarray*}
&& \int_{-1}^1\int_{n_k}^{kn_k}\left|h(x y)\,\e^{iy^2/2}\chi_k\left(\frac{y}{n_k}
\right)\phi''(x)\right|^2\,\mathrm{d}x\,\mathrm{d}y<\varepsilon\,,
\\ &&
\int_{-1}^1\int_{n_k}^{kn_k}\left|\frac{1}{y^2}\,f(x y)\,\e^{iy^2/2}
\chi_k\left(\frac{y}{n_k}\right)\phi''(x)\right|^2\,\mathrm{d}x\,\mathrm{d}y<
\varepsilon\,, \\ &&
\int_{-1}^1\int_{n_k}^{kn_k}\left|\frac{1}{y}\,f'(x y)\,\e^{iy^2/2}\chi_k\left(\frac{y}
{n_k}\right)\phi'(x)\right|^2\,\mathrm{d}x\,\mathrm{d}y<\varepsilon\,.
\end{eqnarray*}
% ------------- %
As for the remaining term in the partial derivative expressions, we simply repeat our calculations from previous section. In this way we are able to conclude that for large enough $k$, and respectively $n_k$ we have
 % ------------- %
\begin{eqnarray*}
&& \int_{\mathbb{R}^2}|H\psi_k|^2\,\mathrm{d}x\,\mathrm{d}y<\|\phi\|^2_{L^\infty
(\mathbb{R})}\int_{n_k}^{kn_k}\int_{-1}^1\biggl|y^2h''(xy)\chi_k\left(\frac{y}
{n_k}\right) \\ && \qquad +f''(xy)\chi_k\left(\frac{y}{n_k}\right) +2ixyh'(xy)\chi_k\left(\frac{y}
{n_k}\right)+ih(xy)\chi_k\left(\frac{y}{n_k}
\right)  \\ && \qquad  -y^2h(xy)\chi_k\left(\frac{y}
{n_k}\right)+\frac{2iy}{n_k}h(xy)\chi_k'\left(\frac{y}{n_k}\right)-f(xy)
\chi_k\left(\frac{y}{n_k}\right) \\ && \qquad -\omega^2\,y^2h(xy)\chi_k\left(\frac{y}{n_k}\right)
-\omega^2\,f(xy)\chi_k\left(\frac{y}{n_k}\right)
+y^2\,V(xy)h(xy)\chi_k\left(\frac{y}{n_k}\right) \\ && \qquad
+V(xy)f(xy)\chi_k\left(\frac{y}{n_k}\right)\biggr|^2\,\mathrm{d}x\,\mathrm{d}y+
\varepsilon  \\ &&
=\|\phi\|^2_{L^\infty(\mathbb{R})}\int_{n_k}^{kn_k}\int_{-1}^1\biggl|y^2\big(h''(xy)
-\omega^2h(xy)+V(xy)h(xy) \\ && \qquad -h(xy)\big)\chi_k\left(\frac{y}{n_k}\right)+ih(xy)\chi_k
\left(\frac{y}{n_k}\right) +f''(xy)\chi_k\left(\frac{y}{n_k}\right) \\ && \qquad +2ixyh'(xy)
\chi_k\left(\frac{y}{n_k}\right)+\frac{2iy}{n_k}h(xy)\chi'_k\left(\frac{y}{n_k}\right)
-f(xy)\chi_k\left(\frac{y}{n_k}\right) \\ && \qquad -\omega^2
\,f(xy)\chi_k\left(\frac{y}{n_k}\right)+V(xy)f(xy)\chi_k\left(\frac{y}{n_k}\right)\biggr|^2\,\mathrm{d}x\,\mathrm{d}y
+\varepsilon\,.
\end{eqnarray*}
 % ------------- %
Using the assumption about the ground state of $\widetilde{L}$, the last equation implies
 % ------------- %
\begin{eqnarray*}
&& \hspace{-2em} \int_{-1}^1\int_{\mathbb{R}}|H\psi_k|^2\,\mathrm{d}x\,\mathrm{d}y<\|\phi\|^2_{L^\infty(\mathbb{R}
)}\int_{n_k}^{kn_k}\int_{-1}^1\biggl|\biggl(f''(xy)+2ixyh'(xy)+ih(xy) \\ &&  \hspace{-2em} -f(xy) -
\omega^2\,f(xy)
+V(xy)f(xy)\biggr)\chi_k\left(\frac{y}{n_k}\right)+\frac{2iy}{n_k}h(xy)\chi_k'
\left(\frac{y}{n_k}\right)\biggr|^2\,\mathrm{d}x\,\mathrm{d}y
+\varepsilon.
\end{eqnarray*}
 % ------------- %
Using the fact that $f(t)=-\frac{i}{2}t^2h(t)$ we conclude in the same way as in the previous section that the right-hand side of the last inequality can be estimated by $9\|\phi\|^2_{L^\infty(\mathbb{R})}\varepsilon$.

The rest of the proof follows the same routine. We pick a sequence $\{\varepsilon_j\}_{j=1}^\infty$ such that $\varepsilon_j\searrow0$ holds as $j\to\infty$ and to any $j$ we construct a function $\psi_{k(\varepsilon_j)}$ with the corresponding numbers chosen in such a way that $n_{k(\varepsilon_j)} >k(\varepsilon_{j-1}) n_{k(\varepsilon_{j-1})}$. The norms of $H\psi_{k(\varepsilon_j)}$ satisfy inequality which (\ref{final}) with $9\|\phi\|^2_{L^\infty(\mathbb{R})}\varepsilon_j$ on the right-hand side, and the sequence $\{\psi_{k(\varepsilon_j)}\}_{j=1}^\infty$ converges weakly to zero by construction, their sequence converges weakly to zero. This proves that $0\in\sigma_{\mathrm{ess}}(H)$; for any nonzero real number $\mu$ we proceed in the same way replacing the above $\psi_k$ with
 % ------------- %
$$
\psi_k(x,y)=h(xy)\,\e^{i\epsilon_\mu(y)} \chi_k\left(\frac{y}{n_k}\right)\phi(x)+
\frac{f(xy)}{y^2}\,\e^{i\epsilon_\mu(y)} \chi_k\left(\frac{y}{n_k}\right)\phi(x)\,,
$$
 % ------------- %
where $\epsilon_\mu(y):= \displaystyle{\int_{\sqrt{|\mu|}}^y\sqrt{t^2+\mu}\,\mathrm{d}t}$, and furthermore, the functions $f,\,\chi_k,\,\phi$ defined in the same way as above.
\end{proof}

Observing the domains of the quadratic form associated with such operators we can extend the result in the following way:

 % ------------- %
\begin{corollary}
The claim of Theorem~\ref{th: interval} remains valid if the Dirichlet boundary conditions at $x=\pm c$ are replaced by any other self-adjoint boundary conditions.
\end{corollary}
 % ------------- %

The result also allows us to answer the question about spectral transition for the model with multiple singular channels.

 % ------------- %
\begin{theorem}
Let $H$ be the operator (\ref{HN}) with the potentials satisfying the stated assumptions, namely the functions $V_j$ are positive with bounded first derivative and $\mathrm{supp}\,V_j \cap \mathrm{supp}\,V_k = \emptyset$ holds for $j\ne k$. Denote by $L_j$ the operator (\ref{comparison}) on $L^2(\R)$ with the potential $V_j$ and $t_V:= \min_j \inf \sigma(L_j)$. Then $H$ is bounded from below if and only if $t_V\ge0$ and in the opposite case its spectrum covers the whole real axis.
\end{theorem}
 % ------------- %
\begin{proof}
The claim follows by bracketing. By assumption we can choose points $x_j$ such that
 % ------------- %
$$
x_0 < v_j^{-} < v_j^{+} < x_1 < v_2^{-} < \cdots < x_{n-1} < v_n^{-} < v_n^{+} < x_n\,,
$$
 % ------------- %
where $v_j^{-}:= \inf\mathrm{supp\,}V_j$ and $v_j^{+}:= \sup\mathrm{supp\,}V_j$ and impose additional Neumann and Dirichlet boundary conditions at them. The spectrum in the intervals $(-\infty,x_0)$ and $(x_n,\infty)$ is found trivially, to the other components of the direct sum obtained in this way we apply Corollary~\ref{c: interval} and Theorem~\ref{th: interval}, respectively.
\end{proof}

%%%%%%%%%%%%%%%%%%%%%%%%%%
\section*{Acknowledgments}
The research was supported by the Czech Science Foundation within the project P203/11/0701. D.B. is grateful to Mittag-Leffler Institute for the hospitality and to H.~Schulz-Baldes for a useful discussion.

%%%%%%%%%%%%%%%%%%%%%%%%%%

\end{document}